\documentclass[10pt,twocolumn]{IEEEtran}

\makeatletter
\def\ps@headings{%
\def\@oddhead{\mbox{}\scriptsize\rightmark \hfil \thepage}%
\def\@evenhead{\scriptsize\thepage \hfil \leftmark\mbox{}}%
\def\@oddfoot{}%
\def\@evenfoot{}}
\makeatother 

\newtheorem{proposition}{Proposition}

\newtheorem{theorem}{Theorem}

\newtheorem{lemma}{Lemma}
\newtheorem{corollary}{Corollary}

\newcommand{\df}{\stackrel{\mbox{\scriptsize def}}{=}}

\newcommand{\rk}{\mathrm{rk}}

\newcommand{\Ar}{A_{\mbox{\tiny{R}}}}

\newcommand{\Ac}{A_{\mbox{\tiny{C}}}}

\newcommand{\ar}{a_{\mbox{\tiny{R}}}}

\newcommand{\dr}{d_{\mbox{\tiny{R}}}}
\newcommand{\ds}{d_{\mbox{\tiny{S}}}}

\newcommand{\di}{d_{\mbox{\tiny{I}}}}

\newcommand{\Jr}{J_{\mbox{\tiny{R}}}}

\newcommand{\Nr}{N_{\mbox{\tiny{R}}}}

\newcommand{\Nc}{N_{\mbox{\tiny{C}}}}

\usepackage{amsmath,amssymb,cite,subfigure}
\usepackage{graphicx}
\usepackage{url}

\begin{document}
\title{Constant-Rank Codes and Their Connection\\ to Constant-Dimension Codes}
\author{Maximilien Gadouleau,~\IEEEmembership{Member,~IEEE,}
and Zhiyuan Yan,~\IEEEmembership{Senior Member,~IEEE}%
\thanks{This work was supported in part by Thales Communications
Inc. and in part by a grant from the Commonwealth of Pennsylvania,
Department of Community and Economic Development, through the
Pennsylvania Infrastructure Technology Alliance (PITA). The work of Zhiyuan Yan was supported in part by a
summer extension grant from Air Force Research Laboratory, Rome, New
York, and the work of Maximilien Gadouleau was supported in part by the ANR project RISC.  Part of the
material in this paper was presented at 2008 IEEE International Symposium on Information Theory (ISIT 2008) in Toronto, Canada and 2008 IEEE International Workshop on Wireless Network Coding (WiNC 2008) in San Francisco, California, USA.} %
\thanks{Maximilien Gadouleau was with the Department of Electrical and Computer Engineering, Lehigh University, Bethlehem, PA 18015 USA. Now he is with CReSTIC, Universit\'e de Reims Champagne-Ardenne, Reims 51100 France. Zhiyuan Yan is with the Department of Electrical and Computer Engineering, Lehigh University, Bethlehem, PA, 18015 USA (e-mail: maximilien.gadouleau@univ-reims.fr; yan@lehigh.edu).}}

\maketitle

\thispagestyle{empty}

\begin{abstract}
Constant-dimension codes have recently received attention due to
their significance to error control in noncoherent random linear network coding. What the maximal cardinality of any constant-dimension code
with finite dimension and minimum distance is and how to construct
the optimal constant-dimension code (or codes) that achieves the
maximal cardinality both remain open research problems. In this
paper, we introduce a new approach to solving these two problems. We
first establish a connection between constant-rank codes and
constant-dimension codes. Via this connection, we show that optimal
constant-dimension codes correspond to optimal constant-rank codes
over matrices with sufficiently many rows. As such, the two
aforementioned problems are equivalent to determining the maximum
cardinality of constant-rank codes and to constructing optimal
constant-rank codes, respectively. To this end, we then derive bounds on
the maximum cardinality of a constant-rank code with a given minimum
rank distance, propose explicit constructions of optimal or
asymptotically optimal constant-rank codes, and establish asymptotic
bounds on the maximum rate of a constant-rank code.
\end{abstract}


\section{Introduction}\label{sec:intro}
While random linear network coding \cite{ho_isit03, chou_allerton03,
ho_it06} has proved to be a powerful tool for disseminating
information in networks, it is highly susceptible to errors caused
by various sources, such as noise, malicious or malfunctioning nodes,
or insufficient min-cut. If received packets are linearly combined
at random to deduce the transmitted message, even a single error in
one erroneous packet could render the entire transmission useless.
Thus, error control for random linear network coding is critical and has
received growing attention recently. Error control schemes proposed
for random linear network coding assume two types of transmission models:
some \cite{cai_itw02, song_it03,yeung_cis06,cai_cis06,zhang_it08}
depend on and take advantage of the underlying network topology or
the particular linear network coding operations performed at various
network nodes; others \cite{koetter_it08,silva_it08} assume that the
transmitter and receiver have no knowledge of such channel transfer
characteristics. The contrast is similar to that between coherent
and noncoherent communication systems.

Error control for noncoherent random linear network coding was first
considered in \cite{koetter_it08}\footnote{A related work
\cite{jaggi_infocom07} considers security issues in noncoherent
random linear network coding.}. Motivated by the property that random linear network coding is vector-space preserving, an operator channel that
captures the essence of the noncoherent transmission model
was defined in \cite{koetter_it08}. Similar to codes defined in
complex Grassmannians for noncoherent multiple-antenna channels,
codes defined in Grassmannians over a finite field
\cite{delsarte_jct76, chihara_siam87}
play a significant role in error control for noncoherent random linear network coding. We refer to these codes as constant-dimension codes
(CDCs) henceforth. These codes can use either the subspace metric
\cite{koetter_it08} or the injection metric
\cite{silva_it09}. The standard advocated approach to random linear network coding (see, e.g., \cite{chou_allerton03}) involves
transmission of packet headers used to record the particular linear
combination of the components of the message present in each
received packet. From coding theoretic perspective, the set of
subspaces generated by the standard approach may be viewed as a
\textbf{suboptimal} CDC with minimum injection distance $1$ in a
Grassmannian, because the whole Grassmannian forms a CDC with
minimum injection distance $1$ \cite{koetter_it08}. Hence, studying
random linear network coding from coding theoretic perspective results in
better error control schemes.

General studies of subspace codes
started only recently (see, for example, \cite{etzion_isit08,
gabidulin_isit08}). On the other hand, there is a steady stream of
works related to codes in Grassmannians. For example, Delsarte
\cite{delsarte_jct76} proved that a Grassmannian endowed with the
injection distance forms an association scheme, and derived its
parameters. The nonexistence of perfect codes in  Grassmannians
was proved in \cite{chihara_siam87,martin_dcc95}. In
\cite{ahlswede_dcc01}, it was shown that Steiner structures yield
diameter-perfect codes in  Grassmannians; properties and
constructions of these structures were studied in
\cite{schwartz_jct02}; in \cite{xia_dcc09}, it was shown that
Steiner structures result in optimal CDCs. Related work on certain
intersecting families and on byte-correcting codes can be found in
\cite{frankl_jct86} and \cite{etzion_it98}, respectively. An
application of codes in Grassmannians to linear authentication
schemes was considered in \cite{wang_it03}. In \cite{koetter_it08},
a Singleton bound for CDCs and a family of codes that are
\textbf{nearly} Singleton-bound-achieving are proposed, a
recursive construction of CDCs which outperform the codes in
\cite{koetter_it08} was given in \cite{skachek_it10}, while a class of codes with even greater cardinality was given in \cite{etzion_it09}. Despite the \textbf{asymptotic optimality} of the Singleton bound and the codes
proposed in \cite{koetter_it08}, neither is optimal in finite cases:
upper bounds tighter than the Singleton bound exist and can be
achieved in some special cases \cite{xia_dcc09}. Thus, two research problems about CDCs remain open: the \textbf{maximal
cardinality} of a CDC with \textbf{finite} dimension and minimum
distance is yet to be determined, and it is not clear how to
construct an optimal code that achieves the maximal cardinality.

In this paper, we introduce a novel approach to solving the two
aforementioned problems. Namely, we aim to solve these problems via
constant-rank codes (CRCs), which
are the counterparts in rank metric codes of constant Hamming weight
codes. There are several reasons for our approach. First, it is
difficult to solve the two problems above directly based on CDCs since
projective spaces lack a natural group structure \cite{silva_it08}.
Also,
the rank metric is very similar to the Hamming metric in many aspects,
and hence familiar results from the Hamming space can be readily
adapted. Furthermore, existing results for rank metric codes in the literature are more extensive than those for CDCs.
Finally, the rank metric has been shown relevant
to error control for both noncoherent \cite{silva_it08} and coherent
\cite{silva_it09} random linear network coding.

Based on our approach, this paper makes two main contributions. Our
first main contribution is that we establish a connection between
CRCs and CDCs. Via this connection, we show that optimal CDCs
correspond to optimal CRCs over matrices with sufficiently many
rows. This connection converts the aforementioned open research
problems about CDCs into research problems about CRCs, thereby
allowing us to take advantage of existing results on rank metric codes in general to tackle such
problems. Despite previous works on rank metric codes, constant-rank codes per se unfortunately have received little attention in the literature. Our second main contribution is our
investigation of CRCs. In particular, we derive
upper and lower bounds on the maximum cardinality of a CRC, propose
explicit constructions of optimal or asymptotically optimal CRCs,
and establish asymptotic bounds on the maximum rate of CRCs. Our investigation of CRCs not only is important for our construction of CDCs, but also serves as a powerful tool to study CDCs and rank metric codes.

The rest of the paper is organized as follows.
Section~\ref{sec:preliminaries} reviews some necessary background.
In Section~\ref{sec:connection}, we determine the connection between
optimal CRCs and optimal CDCs. In Section~\ref{sec:constant_rank},
we study the maximum cardinality of CRCs, and present our results
on the asymptotic behavior of the maximum rate of a CRC.

\section{Preliminaries}\label{sec:preliminaries}
\subsection{Rank metric codes}\label{sec:rank_metric}
Error correction codes with the rank metric \cite{delsarte_jct78,
gabidulin_pit0185, roth_it91} have been receiving steady attention
in the literature due to their applications in storage systems
\cite{roth_it91}, public-key cryptosystems \cite{gabidulin_lncs91},
space-time coding \cite{lusina_it03}, and network coding
\cite{koetter_it08, silva_it08}. Below we review some important
properties of rank metric codes established in \cite{delsarte_jct78,
gabidulin_pit0185, roth_it91}.


For all ${\bf X}, {\bf Y}\in \mathrm{GF}(q)^{m \times n}$, it is
easily verified that $\dr({\bf X},{\bf Y})\df \rk({\bf X} - {\bf
Y})$ is a metric over $\mathrm{GF}(q)^{m \times n}$, referred to as
the \emph{rank metric} henceforth. Please note that the rank metric for the vector representation of rank metric codes is defined differently \cite{gabidulin_pit0185}. Since the connection between the matrix representation of rank metric codes and CDCs is more natural, we consider the matrix representation of rank metric codes henceforth. We
denote the number of matrices of rank $r$ ($0 \leq r \leq
\min\{m,n\}$) in $\mathrm{GF}(q)^{m \times n}$ as $\Nr(q,m,n,r) = {n
\brack r} \alpha(m,r)$ \cite{gabidulin_pit0185}, where $\alpha(m,0)
\df 1$, $\alpha(m,r) \df \prod_{i=0}^{r-1}(q^m-q^i)$, and ${n \brack
r} \df \alpha(n,r)/\alpha(r,r)$ for $r \geq 1$. The term ${n \brack r}$
is often referred to as a Gaussian
binomial~\cite{andrews_book76}, and satisfies
\begin{equation}
    q^{r(n-r)} \leq {n \brack r} < K_q^{-1} q^{r(n-r)}
    \label{eq:Gaussian}
\end{equation}
for all $0 \leq r \leq n$, where $K_q = \prod_{j=1}^\infty
(1-q^{-j})$ \cite{gadouleau_it08_dep}. $K_q^{-1}$ decreases with $q$
and satisfies $1 < K_q^{-1} \leq K_2^{-1} < 4$.
We denote the volume (i.e., the number of points) of the intersection of two \textbf{spheres} in
$\mathrm{GF}(q)^{m \times n}$ of radii $r$ and $s$ and with rank distance $d$ between
their centers as $\Jr(q,m,n,r,s,d)$. A closed-form formula for
$\Jr(q,m,n,r,s,d)$ is determined in \cite{gadouleau_cl09}.

A rank metric code is a subset of $\mathrm{GF}(q)^{m \times n}$, and
its {\em minimum rank distance}, denoted as $\dr$, is simply the
minimum rank distance over all possible pairs of distinct codewords.
It is shown in \cite{delsarte_jct78, gabidulin_pit0185, roth_it91}
that the minimum rank distance of a code of cardinality $M$ in
$\mathrm{GF}(q)^{m \times n}$ satisfies $\dr \leq n-\log_{q^m}M+1.$
In this paper, we refer to this bound as the Singleton bound for
rank metric codes and codes that attain the equality as maximum rank
distance (MRD) codes. We refer to the subclass of MRD codes
introduced in \cite{kshevetskiy_isit05} as generalized Gabidulin
codes. These codes are based on the vector view of rank metric
codes, described as follows. The columns of a matrix ${\bf X} \in
\mathrm{GF}(q)^{m \times n}$ can be mapped into elements of the
field $\mathrm{GF}(q^m)$ according to a basis $B_m$ of
$\mathrm{GF}(q^m)$ over $\mathrm{GF}(q)$. Hence ${\bf X}$ can be
mapped into the vector ${\bf x} \in \mathrm{GF}(q^m)^n$, and the
rank of ${\bf X}$ is equal to the maximum number of linearly
independent coordinates of ${\bf x}$. Generalized Gabidulin codes
are linear MRD codes over $\mathrm{GF}(q^m)$ for $m \geq n$. For all
$q$, $1 \leq d \leq r \leq n \leq m$, the number of codewords of
rank $r$ in an $(n, n-d+1, d)$ linear MRD code over
$\mathrm{GF}(q^m)$ is denoted by $M(q,m,n,d,r)$, and it is known that \cite{gabidulin_pit0185}
\begin{align}
	\nonumber
    &M(q,m,n,d,r)\\
    \label{eq:Mdr_def}
    &= {n \brack r} \sum_{j=d}^r (-1)^{r-j}
    {r \brack j} q^{(r-j)(r-j-1)/2} \left( q^{m(j-d+1)} - 1\right).
\end{align}
We will omit the dependence of the quantities defined above on $q$,
$m$, and $n$ when there is no ambiguity in some proofs.

\subsection{Constant-dimension codes} \label{sec:CDC}

We refer to the set of all subspaces of $\mathrm{GF}(q)^n$ with
dimension $r$ as the Grassmannian of dimension $r$ and denote it as
$E_r(q,n)$, where $|E_r(q,n)| = {n \brack r}$; we refer to $E(q,n) =
\bigcup_{r=0}^n E_r(q,n)$ as the projective space. For $U,V \in
E(q,n)$, their intersection $U \cap V$ is also
a subspace in $E(q,n)$, and we denote the smallest subspace
containing the union of $U$ and $V$ as $U+V$.
Both the \emph{subspace metric} \cite[(3)]{koetter_it08}
$\ds(U,V) \df \dim(U + V) - \dim(U \cap V) = 2\dim(U+V) - \dim(U) -
\dim(V)$ and \emph{injection metric} \cite[Def.~1]{silva_it09}
$\di(U,V) \df \frac{1}{2} \ds(U,V) + \frac{1}{2} |\dim(U) - \dim(V)|
= \dim(U + V) - \min\{\dim(U), \dim(V)\}$ are metrics over $E(q,n)$.

The Grassmannian $E_r(q,n)$ endowed with either the subspace metric
or the injection metric forms an association scheme
\cite{koetter_it08, delsarte_jct76}. Since $\ds(U,V) = 2\di(U,V)$
for all $U,V \in E_r(q,n)$ and  the injection distance provides a
more natural distance spectrum, i.e., $0\leq \di(U,V) \leq r$ for
all $U,V \in E_r(q,n)$, we consider only the injection metric for
Grassmannians and CDCs henceforth.
We denote the number of subspaces in $E_r(q,n)$ at distance $d$ from
a given subspace as $\Nc(d) = q^{d^2} {r \brack d} {n-r \brack d}$
\cite{koetter_it08}.

A subset of $E_r(q,n)$ is called a constant-dimension code (CDC). We
denote the \textbf{maximum} cardinality of a CDC in $E_r(q,n)$ with
minimum distance $d$ as $\Ac(q,n,r,d)$. Constructions of CDCs and
bounds on $\Ac(q,n,r,d)$ have been given in \cite{koetter_it08,
xia_dcc09, skachek_it10, gabidulin_isit08, kohnert_mmics08, etzion_it09}. In
particular, $\Ac(q,n,r,1) = {n \brack r}$ and it is shown
\cite{koetter_it08, xia_dcc09} for $r \leq \left\lfloor
\frac{n}{2}\right\rfloor$ and $2 \leq d \leq r$,
\begin{equation}\label{eq:bounds_Ac}
    q^{(n-r)(r-d+1)} \leq \Ac(q,n,r,d) \leq
    \frac{{n \brack r-d+1}}{{r \brack r-d+1}}.
\end{equation}

\section{Connection between constant-dimension codes and constant-rank
codes}\label{sec:connection}

In this section, we first establish some connections between the
rank metric and the injection metric. We then define constant-rank
codes and we show how optimal constant-rank codes can be used to
construct optimal CDCs.

Let us denote the row space and the column space of ${\bf X} \in
\mathrm{GF}(q)^{m \times n}$ over $\mathrm{GF}(q)$ as $R({\bf X})$
and $C({\bf X})$, respectively. Following the convention of coding theory, a generator matrix of a subspace $U$ is any full rank matrix whose row space is the subspace $U$. The notations introduced above are
naturally extended to codes as follows: for $\mathcal{C} \subseteq
\mathrm{GF}(q)^{m \times n}$, $C(\mathcal{C}) \df \{ U \in E(q,m) :
\exists\, {\bf M} \in \mathcal{C}, C({\bf M}) = U\}$ and $R(\mathcal{C}) \df
\{ V \in E(q,n) : \exists\, {\bf M} \in \mathcal{C}, R({\bf M}) = V\}$.


\begin{lemma}\label{lemma:x_Gamma_Delta}
For $U \in E_r(q,m)$, $V \in E_r(q,n)$, and ${\bf X} \in
\mathrm{GF}(q)^{m \times n}$ with rank $r$, $C({\bf X}) = U$ and
$R({\bf X}) = V$ if and only if there exist a generator matrix ${\bf
G} \in \mathrm{GF}(q)^{r \times m}$ of $U$ and a generator matrix
${\bf H} \in \mathrm{GF}(q)^{r \times n}$ of $V$ such that ${\bf X}
= {\bf G}^T {\bf H}$.
\end{lemma}

The proof of Lemma~\ref{lemma:x_Gamma_Delta} is straightforward and
hence omitted. We remark that ${\bf X} = {\bf G}^T {\bf H}$ is
referred to as a rank factorization \cite{rao_book00}.
We now derive a relation between the rank distance between two
matrices and the injection distances between their respective row
and column spaces.

\begin{theorem}\label{th:x_y_E_F}
For all ${\bf X}, {\bf Y} \in \mathrm{GF}(q)^{m \times n}$,
\begin{align*} 
    &\di(R({\bf X}), R({\bf Y})) +
    \di(C({\bf X}), C({\bf Y})) - |\rk({\bf X}) - \rk({\bf
    Y})|\\
    &\leq \dr({\bf X}, {\bf Y})\\
    & \leq \min\{\di(R({\bf X}), R({\bf Y})),
    \di(C({\bf X}), C({\bf Y}))\}\\
    &+ \min\{\rk({\bf X}), \rk({\bf Y})\}.
\end{align*}
\end{theorem}

\begin{proof}
By Lemma~\ref{lemma:x_Gamma_Delta}, we have ${\bf X} = {\bf C}^T
{\bf R}$ and ${\bf Y} = {\bf D}^T {\bf S}$, where ${\bf C} \in
\mathrm{GF}(q)^{\rk({\bf X}) \times m}$, ${\bf R} \in
\mathrm{GF}(q)^{\rk({\bf X}) \times n}$, ${\bf D} \in
\mathrm{GF}(q)^{\rk({\bf Y}) \times m}$, ${\bf S} \in
\mathrm{GF}(q)^{\rk({\bf Y}) \times n}$ are generator matrices of
$C({\bf X})$, $R({\bf X})$, $C({\bf Y})$, and $R({\bf Y})$,
respectively. Hence ${\bf X} - {\bf Y} = ({\bf C}^T | -{\bf D}^T)
({\bf R}^T | {\bf S}^T)^T$ and $\rk({\bf X} - {\bf Y}) \leq
\min\{ \rk({\bf C}^T | -{\bf D}^T), \rk({\bf R}^T | {\bf S}^T) \}$. Sylvester's law of nullity in \cite[Corollary 6.1]{marsaglia_lma74}
or in \cite[0.4.5 (c)]{horn_book85}, states that $\rk({\bf A}{\bf B}) \geq \rk({\bf A}) + \rk({\bf B}) - n$ for any matrices ${\bf A}$ with $n$ columns and ${\bf B}$ with $n$ rows. Therefore,
\begin{align*}
    &\rk({\bf C}^T | -{\bf D}^T) + \rk({\bf R}^T | {\bf S}^T) - \rk({\bf
    X}) - \rk({\bf Y})\\
    &\leq \rk({\bf X} - {\bf Y})\\
    &\leq \min\{ \rk({\bf C}^T | -{\bf D}^T), \rk({\bf R}^T | {\bf S}^T) \}.
\end{align*}
Since $\rk({\bf C}^T | -{\bf D}^T) = \di(C({\bf X}), C({\bf Y})) +
\min\{\rk({\bf X}), \rk({\bf Y})\}$ and $\rk({\bf R}^T | {\bf S}^T)
= \di(R({\bf X}), R({\bf Y})) + \min\{\rk({\bf X}), \rk({\bf Y})\}$,
we obtain the claim.
\end{proof}

A {\em constant-rank code} (CRC) of constant rank $r$ in
$\mathrm{GF}(q)^{m \times n}$ is a nonempty subset of
$\mathrm{GF}(q)^{m \times n}$ such that all elements have rank $r$.
Proposition~\ref{prop:C_R(C)_C(C)} below shows how a CRC leads to
two CDCs with their minimum injection distance related to the
minimum rank distance of the CRC.

\begin{proposition}\label{prop:C_R(C)_C(C)}
Let $\mathcal{C}$ be a CRC of constant rank $r$ and minimum distance
$\dr$ in $\mathrm{GF}(q)^{m \times n}$. Then $R(\mathcal{C})
\subseteq E_r(q,n)$ and $C(\mathcal{C}) \subseteq E_r(q,m)$ have
minimum distances at least $\dr-r$.
\end{proposition}

Proposition~\ref{prop:C_R(C)_C(C)} follows directly from
Theorem~\ref{th:x_y_E_F} and hence its proof is omitted. When the minimum rank
distance of a CRC is greater than its constant rank,
Proposition~\ref{prop:E(C)} below shows how the CRC leads to two
CDCs with the \textbf{same cardinality}, and the relations between
their distances can be further strengthened.

\begin{proposition}\label{prop:E(C)}
If $\mathcal{C}$ is a CRC of constant rank $r$ and minimum rank distance
$d+r$ ($1 \leq d \leq r$) in $\mathrm{GF}(q)^{m \times n}$, then
$R(\mathcal{C}) \subseteq E_r(q,n)$ and $C(\mathcal{C}) \subseteq
E_r(q,m)$ have cardinality $|\mathcal{C}|$ and their minimum injection
distances satisfy $\di(C(\mathcal{C})) + \di(R(\mathcal{C})) \leq
d+r \leq \min\{\di(C(\mathcal{C})), \di(R(\mathcal{C}))\} + r$.
\end{proposition}

\begin{proof}
Let ${\bf X}$ and ${\bf Y}$ be \textbf{any} two distinct codewords
in $\mathcal{C}$. By Theorem~\ref{th:x_y_E_F}, $\di(R({\bf X}),
R({\bf Y})) \geq \dr({\bf X}, {\bf Y}) - r \geq d
> 0$, and hence $\di(R(\mathcal{C})) \geq d$ and $|R(\mathcal{C})| = |\mathcal{C}|$.
Similarly, $\di(C({\bf X}), C({\bf Y})) \geq d
> 0$, and thus $\di(C(\mathcal{C})) \geq d$ and $|C(\mathcal{C})| =
|\mathcal{C}|$. Furthermore, if $\dr({\bf X}, {\bf Y}) = d+r$, then
by Theorem~\ref{th:x_y_E_F}, $d+r \geq \di(C({\bf X}),C({\bf Y})) +
\di(R({\bf X}), R({\bf Y})) \geq \di(C(\mathcal{C})) +
\di(R(\mathcal{C}))$.
\end{proof}

We remark that the requirement of having a minimum distance greater than the constant rank is a strong condition on the CRC. Indeed, any codeword of a linear code has rank at least equal to the minimum distance of a code. Therefore, no set of codewords of a linear code (and, in particular, a linear MRD code) satisfies this condition. Therefore, while CRCs with minimum distance no more than their constant-rank will be directly constructed from linear MRD codes in Section \ref{sec:constructive_bounds}, designing CRCs with minimum distance greater than their constant-rank will require translates of codes instead, which are not as easy to manipulate.

Propositions \ref{prop:C_R(C)_C(C)} and \ref{prop:E(C)} show how to
construct CDCs from a CRC. Alternatively, Proposition
\ref{prop:C_Gamma_Delta} below shows that we can construct a CRC
from a pair of CDCs.

\begin{proposition}\label{prop:C_Gamma_Delta}
Let $\mathcal{M}$ be a CDC in $E_r(q,m)$ and $\mathcal{N}$ be a CDC
in $E_r(q,n)$ such that $|\mathcal{M}| = |\mathcal{N}|$. Then there
exists a CRC $\mathcal{C} \subseteq \mathrm{GF}(q)^{m \times n}$
with constant rank $r$ and cardinality $|\mathcal{M}|$ satisfying
$C(\mathcal{C}) = \mathcal{M}$ and $R(\mathcal{C}) = \mathcal{N}$.
Furthermore, its minimum distance $\dr$ satisfies $\di(\mathcal{N})
+ \di(\mathcal{M}) \leq \dr \leq \min\{\di(\mathcal{N}),
\di(\mathcal{M})\} + r$.
\end{proposition}

\begin{proof}
Denote the generator matrices of the component subspaces of
$\mathcal{M}$ and $\mathcal{N}$ as ${\bf G}_i$ and ${\bf H}_i$,
respectively and define the code $\mathcal{C}$ formed by the
codewords ${\bf X}_i = {\bf G}_i^T {\bf H}_i$ for $0 \leq i \leq
|\mathcal{M}|-1$. Then $C(\mathcal{C}) = \mathcal{M}$ and
$R(\mathcal{C}) = \mathcal{N}$ by Lemma~\ref{lemma:x_Gamma_Delta}
and the lower bound on $\dr$ follows from Theorem~\ref{th:x_y_E_F}. Let
${\bf X}_i$ and ${\bf X}_j$ be distinct codewords in $\mathcal{C}$
such that $\di(C({\bf X}_i), C({\bf X}_j)) = \di(\mathcal{M})$. By
Theorem~\ref{th:x_y_E_F}, we obtain $\dr \leq \dr({\bf X}_i, {\bf
X}_j) \leq \di(\mathcal{M}) + r$. Similarly, we also obtain $\dr
\leq \di(\mathcal{N}) + r$.
\end{proof}

The connections between general CRCs and CDCs derived above
naturally imply relations between optimal CRCs and optimal CDCs. We
denote the maximum cardinality of a CRC in $\mathrm{GF}(q)^{m \times
n}$ with constant rank $r$ and minimum rank distance $d$ as
$\Ar(q,m,n,d,r)$. If $\mathcal{C}$ is a CRC in $\mathrm{GF}(q)^{m
\times n}$ with constant rank $r$, then its transpose code
$\mathcal{C}^T$ forms a CRC in $\mathrm{GF}(q)^{n \times m}$ with
the same constant rank, minimum distance, and cardinality. Therefore
$\Ar(q,m,n,d,r) = \Ar(q,n,m,d,r)$, and henceforth in this paper we
assume $n \leq m$ without loss of generality. We further observe
that $\Ar(q,m,n,d,r)$ is a non-decreasing function of $m$ and $n$,
and a non-increasing function of $d$,  and that $\Ac(q,n,r,d)$ is a
non-decreasing function of $n$ and a non-increasing function of $d$.

\begin{proposition}\label{prop:A_As}
For all $q$, $1 \leq d \leq r \leq n \leq m$, and any $0 \leq p \leq
r$,
\begin{align}
	\nonumber 
    &\min\{\Ac(q,n,r,d+p), \Ac(q,m,r,r-p)\}\\
    \label{eq:A_As1}
    &\leq \Ar(q,m,n,d+r,r) \leq \Ac(q,n,r,d).
\end{align}
\end{proposition}

\begin{proof}
Using the monotone properties of $\Ar(q,m,n,\dr,r)$ and
$\Ac(q,n,r,d)$ above, the upper bound follows from
Proposition~\ref{prop:E(C)}, while the lower bound follows from
Proposition~\ref{prop:C_Gamma_Delta} for $\di(\mathcal{M}) = r-p$
and $\di(\mathcal{N}) = d+p$.
\end{proof}

We remark that the lower bound in (\ref{eq:A_As1}) is trivial for
$d+p > \min\{r, n-r\}$ or $r-p > \min\{r, m-r\}$. Therefore, the
lower bound in (\ref{eq:A_As1}) is nontrivial when $\max\{0, 2r-m\}
\leq p \leq \min\{r-d, n-r-d\}$.

Combining the bounds in (\ref{eq:A_As1}), we obtain that the
cardinalities of optimal CRCs over matrices with sufficiently many rows
equal the cardinalities of CDCs with related distances.
Furthermore, we show that optimal CDCs can be constructed from such
optimal CRCs.

\begin{theorem}\label{th:A=As}
For all $q$, $2r \leq n \leq m$, and $1 \leq d \leq r$,
$\Ar(q,m,n,d+r,r) = \Ac(q,n,r,d)$ if either $d=r$ or $m \geq m_0$,
where $m_0 = (n-r)(r-d+1) + r + 1$. Furthermore, if $\mathcal{C}$ is
an optimal CRC in $\mathrm{GF}(q)^{m \times n}$ with constant rank $r$ and minimum distance $d+r$ for $m \geq m_0$ or $d=r$, then
$R(\mathcal{C})$ is an optimal CDC in $E_r(q,n)$ with minimum
distance $d$.
\end{theorem}

\begin{proof}
First, the case where $d=r$ directly follows from (\ref{eq:A_As1}) for
$p=0$. Second, if $d < r$ and $m \geq m_0$, by~(\ref{eq:bounds_Ac})
we obtain $\Ac(q,m,r,r) \geq q^{m-r} \geq q^{m_0-r}$. Also, by
\cite[Lemma 1]{gadouleau_it08_dep}, we obtain $q^{r(r-d+1)-1} <
\alpha(r,r-d+1) \leq q^{r(r-d+1)}$ for all $2 \leq d < r$, and
hence~(\ref{eq:bounds_Ac}) yields $\Ac(q,n,r,d) < q^{(n-r)(r-d+1) +
1} = q^{m_0-r} \leq \Ac(q,m,r,r)$. Thus, when $p=0$, the lower bound
in~(\ref{eq:A_As1}) simplifies to $\Ar(q,m,n,d+r,r) \geq
\Ac(q,n,r,d)$. Combining with the upper bound in~(\ref{eq:A_As1}),
we obtain $\Ar(q,m,n,d+r,r) = \Ac(q,n,r,d)$.

The second claim immediately follows from Proposition~\ref{prop:E(C)}.
\end{proof}
Theorem~\ref{th:A=As} implies that to determine $\Ac(q,n,r,d)$ and
to construct optimal CDCs, it is sufficient to determine
$\Ar(q,m,n,d+r,r)$ and to construct optimal CRCs over matrices with
sufficiently many rows. We observe that this implies that
$\Ar(q,m,n,d+r,r)$ remains constant for all $m \geq m_0$. When
$d=r$, $\Ar(q,m,n,2r,r)$ remains constant for $m \geq n$. When
$d=1$, $m_0 = (n-r+1)r + 1$, but $\Ar(q,m,n,r+1,r)$ remains constant
for $m \geq n$, and this is shown in
Section~\ref{sec:constructive_bounds}.

In comparison to existing constructions of CDCs \cite{koetter_it08, silva_it08,
skachek_it10, etzion_isit08, xia_dcc09, kohnert_mmics08}, our construction based on CRCs has two advantages.
\textbf{First} and foremost, by Theorem~\ref{th:A=As}, our construction leads to optimal
CDCs for all parameter values.
In contrast, none of previously proposed constructions lead to optimal CDCs for all parameter values. For example, the construction based on liftings of rank metric codes \cite{koetter_it08,
silva_it08} leads to suboptimal CDCs (though sometimes they may be nearly optimal). This is because CDCs of dimension $r$ based on liftings of rank metric codes have the highest possible covering radius $r$
\cite{gadouleau_it09_cdc}, which implies there exists a subspace that can be added to such CDCs
without decreasing the minimum distance.
The CDCs constructed in similar approaches \cite{skachek_it10} are not optimal for the same reason. The optimality for some constructions \cite{etzion_isit08, etzion_it09} are not clear.  The
construction based on Steiner structures \cite{xia_dcc09} and that based on computational techniques \cite{kohnert_mmics08} lead to optimal CDCs, but are applicable to special cases only.
The \textbf{second} advantage of our construction is
an additional degree of freedom, which is the number $m$ of rows of
the matrices. By Theorem \ref{th:A=As}, optimal CRCs lead to optimal
CDCs provided that $m \geq m_0$, and hence the parameter $m$ may vary
anywhere above the lower bound $m_0$. On the other hand, the
constructions in the literature use fixed dimensions and do not
introduce any new parameter. For instance, in order to obtain a CDC
in $E_r(q,n)$ by lifting a rank metric code, the original code must
be in $\mathrm{GF}(q)^{r \times (n-r)}$. This additional degree of
freedom is significant for code design, as it may be easier to
construct optimal CRCs with larger $m$. Thus our construction is a very promising approach to solving the two open research problems mentioned in Section~\ref{sec:intro}.

\section{Constant-rank codes}\label{sec:constant_rank}

Having proved that optimal CRCs over matrices with sufficiently many
rows lead to optimal CDCs, in this section we investigate the
properties of CRCs.

\subsection{Bounds}

We now derive bounds on the maximum cardinality of CRCs. We first
remark that the bounds on $\Ar(q,m,n,d,r)$ derived in
Section~\ref{sec:connection} can be used in this section. Also,
since $\Ar(q,m,n,1,r) = \Nr(q,m,n,r)$ and $\Ar(q,m,n,d,r) = 1$ for
$d > 2r$, we shall assume $2 \leq d \leq 2r$ henceforth\footnote{Since the minimum distance of a code is defined using pairs of distinct codewords, the minimum distance for a code of cardinality one is defined to be zero sometimes.}.

We first derive the counterparts of the Gilbert and the Hamming
bounds for CRCs in terms of intersections of spheres with rank
radii.

%

\begin{proposition}\label{prop:hamming_general}
For all $q$, $1 \leq r,d \leq n \leq m$, and $t =\lfloor
\frac{d-1}{2} \rfloor$,
\begin{align*}
    &\frac{\Nr(q,m,n,r)}{\sum_{i=0}^{d-1} \Jr(q,m,n,i,r,r)}\\
    &\leq \Ar(q,m,n,d,r)\\ 
    &\leq \min_{1 \leq s \leq n} \left\{
    \frac{\Nr(q,m,n,s)}{\sum_{i=0}^t \Jr(q,m,n,i,s,r)} \right\}.
\end{align*}
\end{proposition}

\begin{proof}
The proof of the lower bound is straightforward and hence omitted.
Let $\mathcal{C} = \{{\bf c}_k\}_{k=0}^{K-1}$ be a CRC with
constant rank $r$ and minimum distance $d$ in $\mathrm{GF}(q)^{m
\times n}$. For all $0 \leq k \leq K-1$ and $0 \leq s \leq n-1$, if
we denote the set of matrices in $\mathrm{GF}(q)^{m \times n}$ with
rank $s$ and distance $\leq t$ from ${\bf c}_k$ as $R_{k,s}$, then
$|R_{k,s}| = \sum_{i=0}^t \Jr(i,s,r)$ for all $k$. Clearly $R_{k,s}
\cap R_{l,s} = \emptyset$ for all $k \neq l$, and hence $\Nr(s) \geq
|\bigcup_{k=0}^{K-1} R_{k,s}| = K |R_{k,s}|$, which yields the upper
bound.
\end{proof}

We now derive upper bounds on $\Ar(q,m,n,d,r)$. We begin by proving
the counterpart in rank metric codes of a well-known bound on
constant-weight codes proved by Johnson in \cite{johnson_it62}.

\begin{proposition}[Johnson bound for rank metric codes]\label{prop:johnson}
For all $q$, $1 \leq r,d < n \leq m$, $\Ar(q,m,n,d,r) \leq
\frac{q^n-1}{q^{n-r}-1} \Ar(q,m,n-1,d,r)$.
\end{proposition}

\begin{proof}
Let $\mathcal{C}$ be an optimal CRC in $\mathrm{GF}(q)^{m \times n}$
with constant rank $r$ and minimum distance $d$. For all ${\bf C}
\in \mathcal{C}$ and all $V \in E_{n-1}(q,n)$, we define $f(V,{\bf
C}) = 1$ if $R({\bf C}) \subseteq V$ and $f(V, {\bf C}) = 0$
otherwise. For any ${\bf C}$, the row space of ${\bf C}$ is
contained in ${n-r \brack 1}$ subspaces in $E_{n-1}(q,n)$ and hence
$\sum_{V \in E_{n-1}(q,n)} f(V, {\bf C}) = {n-r \brack 1}$; for all
$V$, $\sum_{{\bf C} \in \mathcal{C}} f(V,{\bf C}) = |\{ {\bf C} \in
C : R({\bf C}) \subseteq V \}|$. Summing over all possible pairs, we
obtain
\begin{align*}
    &\sum_{V \in E_{n-1}(q,n)} \sum_{{\bf C} \in \mathcal{C}} f(V,{\bf
    C})\\
    &= \sum_{V \in E_{n-1}(q,n)} |\{ {\bf C} \in C : R({\bf C}) \subseteq V \}|,\\
    &\sum_{{\bf C} \in \mathcal{C}} \sum_{V \in E_{n-1}(q,n)}  f(V,{\bf
    C})= {n-r \brack 1} \Ar(q,m,n,d,r).
\end{align*}
Hence there exists $U \in E_{n-1}(q,n)$ such that $|\{ {\bf C} \in C
: R({\bf C}) \subseteq U \}| = \sum_{{\bf C} \in \mathcal{C}}
f(U,{\bf C}) \geq \frac{{n-r \brack 1}}{{n \brack 1}}
\Ar(q,m,n,d,r)$. By Lemma~\ref{lemma:x_Gamma_Delta}, all the
codewords ${\bf C}_i$ with $R({\bf C}_i) \subseteq U$ can be
expressed as ${\bf C}_i = {\bf G}_i^T {\bf H}_i {\bf U}$, where
${\bf H}_i \in \mathrm{GF}(q)^{r \times (n-1)}$ and ${\bf U} \in
\mathrm{GF}(q)^{(n-1) \times n}$ is a generator matrix of $U$.
Therefore, the code $\{ {\bf G}_i^T {\bf H}_i \}$ forms a CRC in
$\mathrm{GF}(q)^{m \times (n-1)}$ with constant rank $r$, minimum
distance $d$, and cardinality $|\{ {\bf C} \in C : R({\bf C})
\subseteq U \}|$, and hence $\frac{q^{n-r} - 1}{q^n-1}
\Ar(q,m,n,d,r) \leq |\{ {\bf C} \in C : R({\bf C}) \subseteq U \}|
\leq \Ar(q,m,n-1,d,r)$.
\end{proof}

The Singleton bound for rank metric codes yields upper bounds on
$\Ar(q,m,n,d,r)$. For any $I \subseteq \{0, 1, \ldots, n\}$, let
$\Ar(q,m,n,d,I)$ denote the maximum cardinality of a code in
$\mathrm{GF}(q)^{m \times n}$ with minimum rank distance $d$ such
that all codewords have ranks belonging to $I$. Then $\Ar(q,m,n,d,r)
\leq q^{m(n-d+1)} - \Ar(q,m,n,d,P_r)$, where $P_r \df \{ i \,:\, 0
\leq i \leq n, |i-r| \geq d \}$. We now determine the counterpart of
the Singleton bound for CRCs.

\begin{proposition}[Singleton bound for CRCs]\label{prop:singleton}
For all $0 \leq i \leq \min\{d-1,r\}$, $\Ar(q,m,n,d,r) \leq
\Ar(q,m,n-i,d-i,J_i)$, where $J_i =\{r-i, r-i+1,
\ldots, \min\{n-i,r\} \}$.
\end{proposition}

\begin{proof}
Let $\mathcal{C}$ be an optimal CRC in $\mathrm{GF}(q)^{m \times n}$
with constant rank $r$ and minimum distance $d$, and consider the
code $\mathcal{C}_i$ obtained by puncturing $i$ coordinates of the
codewords in $\mathcal{C}$. Since $i \leq r$, the codewords of
$\mathcal{C}_i$ all have ranks between $r-i$ and $\min\{n-i,r\}$.
Also, since $i < d$, any two codewords have distinct puncturings,
and we obtain $|\mathcal{C}_i| = |\mathcal{C}|$ and
$\dr(\mathcal{C}_i) \geq d-i$. Hence $\Ar(q,m,n,d,r) = |\mathcal{C}|
= |\mathcal{C}_i| \leq \Ar(q,m,n-i,d-i,J_i)$.
\end{proof}

We now combine the counterpart of the Johnson bound in
Proposition~\ref{prop:johnson} and that of the Singleton bound in
Proposition~\ref{prop:singleton} in order to obtain an upper bound
on $\Ar(q,m,n,d,r)$ for $d \leq r$.

\begin{proposition}\label{prop:singleton_johnson}
For all $q$, $1 \leq d \leq r \leq n \leq m$, $\Ar(q,m,n,d,r) \leq
{n \brack r} \alpha(m,r-d+1)$.
\end{proposition}

\begin{proof}
Applying Proposition~\ref{prop:johnson} $n-r$ times successively, we
obtain $\Ar(q,m,n,d,r) \leq {n \brack r} \Ar(q,m,r,d,r)$. For $n=r$
and $i = d-1$, $J_i = \{r-d+1\}$ and hence Proposition
\ref{prop:singleton} yields $\Ar(q,m,r,d,r) \leq
\Ar(q,m,r-d+1,1,r-d+1) = \Nr(q,m,r-d+1,r-d+1) = \alpha(m,r-d+1)$.
Thus $\Ar(q,m,n,d,r) \leq {n \brack r} \alpha(m,r-d+1)$.
\end{proof}

We now derive the counterpart in rank metric codes of the
Bassalygo-Elias bound \cite{bassalygo_pit65} and we also tighten the
bound when $d > r+1$. For a code $\mathcal{C} \subseteq \mathrm{GF}(q)^{l \times k}$ ($k \leq l $), $A_i \df |\{ {\bf C} \in \mathcal{C}: \rk({\bf C}) = i \}|$ for $0\leq i \leq l$; we refer to $A_i$'s as the rank distribution of $\mathcal{C}$.

\begin{proposition}[Bassalygo-Elias bound for rank metric codes]\label{prop:bassalygo_s}
For $\max\{r,d\} \leq k \leq n$, $0 \leq s \leq k$, $k \leq l \leq
m$, and any code $\mathcal{C} \subseteq \mathrm{GF}(q)^{l \times k}$
with minimum rank distance $d$ and rank distribution $A_i$'s,
\begin{equation}\label{eq:bassalygo_s}
    \Ar(q,m,n,d,r) \geq \max_{s, \left\{A_i\right\}, k, l}\frac{\sum_{i=0}^n A_i
    \Jr(q,l,k,s,r,i)}{\Nr(q,l,k,s)}.
\end{equation}
Furthermore, if $r+1< d \leq 2r$, then
\begin{align}
	\nonumber
    &\Ar(q,m,n,d,r)\\
    \label{eq:bassalygo_extended_s}
    &\geq \max_{s, \left\{A_i\right\}, k, l} \frac{\sum_{i=0}^n A_i
    \Jr(q,l,k,s,r,i)}{\Nr(q,l,k,s) - \sum_{i=0}^n A_i \sum_{t=0}^{d-r-1}
    \Jr(q,l,k,s,t,i)}.
\end{align}
\end{proposition}

The proof of Proposition~\ref{prop:bassalygo_s} is given in
Appendix~\ref{app:prop:bassalygo_s}.

Although the RHS of (\ref{eq:bassalygo_s})
and~(\ref{eq:bassalygo_extended_s}) can be maximized over $\{A_i\}$,
it is difficult to do so since $\{A_i\}$ is not available for most
rank metric codes with the exception of linear MRD codes. Thus, we
derive a bound using the rank weight distribution of linear MRD
codes.

\begin{corollary}\label{cor:bassalygo}
For all $q$, $1 \leq r,d \leq n \leq m$, $\Ar(q,m,n,d,r) \geq
\Nr(q,m,n,r) q^{m(-d+1)}$.
\end{corollary}

\begin{proof}
Applying (\ref{eq:bassalygo_s}) to an $(n,n-d+1,d)$ MRD code over
$\mathrm{GF}(q^m)$, we obtain $\Nr(s) \Ar(d,r) \geq
\sum_{i=0}^n M(d,i) \Jr(s,r,i)$. Summing for all $0 \leq
s \leq n$, we obtain $\Ar(d,r) \geq \Nr(r) q^{m(-d+1)}$
since $\sum_{s=0}^n \Jr(s,r,i) = \Nr(r)$.
\end{proof}

The RHS of~(\ref{eq:bassalygo_s})
and~(\ref{eq:bassalygo_extended_s}) decrease rapidly with increasing
$d$, rendering the bounds in (\ref{eq:bassalygo_s})
and~(\ref{eq:bassalygo_extended_s}) trivial for $d$ approaching $2r$.

Proposition \ref{prop:bassalygo_vs_johnson} below shows that the
bound in Corollary \ref{cor:bassalygo} is tight up to a scalar for
$d \leq r$. To measure the tightness, we introduce a ratio $C(q,m,n,d,r) \df
\Ar(q,m,n,d,r)/[\Nr(q,m,n,r)q^{m(-d+1)}]$ for $2 \leq d \leq r \leq n \leq m$.

\begin{proposition}\label{prop:bassalygo_vs_johnson}
For all $q$, $2 \leq d \leq r \leq n \leq m$, $C(q,m,n,d,r) \leq
\frac{q^2}{q^2-1}$ for $r+d-1 \leq m$ and $C(q,m,n,d,r) <
\frac{q-1}{q}K_q^{-1}$ otherwise.
\end{proposition}

\begin{proof}
By Proposition~\ref{prop:singleton_johnson}, $C(d,r) \leq
q^{m(d-1)} \alpha(m,r-d+1)/\alpha(m,r) =
q^{(m-r+d-1)(d-1)}/\alpha(m-r+d-1,d-1)$. Since $\alpha(n,l) >
\frac{q}{q-1} K_q q^{nl}$ for all $1 \leq l \leq n-1$ \cite[Lemma
1]{gadouleau_it08_dep}, we obtain $C(d,r) <
\frac{q-1}{q}K_q^{-1}$. Finally, $\alpha(n,l) \geq \frac{q^2-1}{q^2}
q^{nl}$ for $l \leq n-l$ \cite[Lemma 1]{gadouleau_it08_dep} yields
$C(d,r) \leq \frac{q^2}{q^2-1}$ for $r+d-1 \leq m$.
\end{proof}

The proof of Proposition \ref{prop:bassalygo_vs_johnson} indicates
that the upper bound in Proposition \ref{prop:singleton_johnson} is
also tight up to a scalar for $d \leq r$. However, these bounds are
not constructive. Below we derive constructive bounds on
$\Ar(q,m,n,d,r)$.

\subsection{Constructions of CRCs}\label{sec:constructive_bounds}

We first give a construction of asymptotically optimal CRCs when $d
\leq r$. We assume the matrices in $\mathrm{GF}(q)^{m \times n}$ are
mapped into vectors in $\mathrm{GF}(q^m)^n$ according to a fixed
basis $B_m$ of $\mathrm{GF}(q^m)$ over $\mathrm{GF}(q)$.

\begin{proposition}\label{prop:Mdr}
For all $q$, $2 \leq d \leq r \leq n \leq m$, $\Ar(q,m,n,d,r) \geq
M(q,m,n,d,r) > {n \brack r} q^{m(r-d)}$.
\end{proposition}

\begin{proof}
The codewords of rank $r$ in an $(n, n-d+1,d)$ linear MRD code over
$\mathrm{GF}(q^m)$ form a CRC in $\mathrm{GF}(q)^{m \times n}$ with
constant rank $r$ and minimum distance $d$. Thus, $\Ar(d,r)
\geq M(d,r)$.

We now prove the lower bound on $M(d,r)$. First, for $d=r$,
$M(r,r) = {n \brack r} (q^m-1) > {n \brack r}$. Second,
suppose $d < r$. By~(\ref{eq:Mdr_def}), $M(d,r)$ can be
expressed as $M(d,r) = {n \brack r}\sum_{j=d}^r (-1)^{r-j}
\mu_j$, where $\mu_j \df q^{(r-j)(r-j-1)/2} {r \brack j}
(q^{m(j-d+1)}-1)$. It can be easily shown that $\mu_j > \mu_{j-1}$
for $d+1 \leq j \leq r$, and hence $M(d,r) \geq {n \brack r}
(\mu_r - \mu_{r-1})$. Therefore, $M(d,r) \geq {n \brack r}
[(q^{m(r-d+1)}-1) - {r \brack 1}(q^{m(r-d)}-1)] > {n \brack r}
q^{m(r-d)}.$
\end{proof}

\begin{corollary}\label{cor:A_r}
For all $q$, $1 \leq r \leq n \leq m$, $\Ar(q,m,n,r,r) = {n \brack
r} (q^m-1)$.
\end{corollary}

\begin{proof}
By Proposition~\ref{prop:singleton_johnson}, $\Ar(r,r) \leq {n
\brack r} (q^m-1)$, and by Proposition~\ref{prop:Mdr},
$\Ar(r,r) \geq M(r,r) = {n \brack r} (q^m-1)$.
\end{proof}

By Corollary~\ref{cor:A_r}, the codewords of rank $r$ in an $(n,
n-r+1,r)$ linear MRD code are \textbf{optimal} CRCs with minimum distance
$r$. Proposition~\ref{prop:bound_Mdr} shows that for all but one
case, the codewords of rank $r$ in an $(n,n-d+1,d)$ MRD code form a
code whose cardinality is close to that of an optimal CRC up to a
scalar which tends to $1$ for large $q$. To measure the optimality, we introduce a ratio
$B(q,m,n,d,r) \df
\Ar(q,m,n,d,r)/M(q,m,n,d,r)$ for  $1 \leq d < r \leq n \leq m$.
\begin{proposition}\label{prop:bound_Mdr}
For all $q$, $1 \leq d < r \leq n \leq m$ and  $m \geq 3$,
\begin{align}
    &B(2,m,m,m-1,m) \leq 2^{m-1}-1 \label{eq:B_1}\\
    &B(q,m,m,m-1,m) < \frac{q-1}{q-2} \quad \mbox{for}\,\, q>2 \label{eq:B_2}\\
    &B(q,m,m,m-2,m) < \frac{(q^2-1)(q-1)}{(q^2-1)(q-2)+1} \label{eq:B_3}\\
	\nonumber
    &B(q,m,m,d,m)\\
    &< \frac{(q^3-1)(q^2-1)(q-1)}{(q^3-1)(q^2-1)(q-2) + q^3-2} \quad \mbox{for}\, d < m-2 \label{eq:B_4}\\
    &B(q,m,n,d,r) < \frac{q}{q-1} \quad \mbox{for}\,\, r<m. \label{eq:B_5}
\end{align}
\end{proposition}

The proof of Proposition~\ref{prop:bound_Mdr} is given in
Appendix~\ref{app:prop:bound_Mdr}.

We now construct CRCs for $d > r$ using generalized Gabidulin codes
\cite{kshevetskiy_isit05}. Let ${\bf g} \in \mathrm{GF}(q^m)^n$ have
rank $n$, and for $0 \leq i \leq m-1$, denote the vector in
$\mathrm{GF}(q^m)^n$ obtained by raising each coordinate of ${\bf
g}$ to the $q^{ai}$-th power, ${\bf g}^{[i]}$, where $a$ and $m$
are coprime. Let $\mathcal{C}$ be the $(n,n-d+1,d)$ generalized
Gabidulin code over $\mathrm{GF}(q^m)$ generated by $\left({\bf
g}^{[0]^T}, \,\, {\bf g}^{[1]^T}, \,\, \ldots, \,\, {\bf
g}^{[n-d]^T} \right)^T$, and $\mathcal{C}'$ be the $(n,d-r,n-d+r+1)$
generalized Gabidulin code generated by $\left({\bf g}^{[n-d+1]^T},
\,\, {\bf g}^{[n-d+2]^T}, \,\, \ldots, \,\, {\bf g}^{[n-r]^T}
\right)^T$. We consider the coset $\mathcal{C} + {\bf c}'$, where
${\bf c}' \in \mathcal{C}'$, and we denote the number of codewords
of rank $r$ in $\mathcal{C} + {\bf c}'$ as $\sigma_r({\bf c}')$.

\begin{lemma}\label{lemma:c'}
For all $d > r$, there exists ${\bf c}' \in \mathcal{C}'$ such that
$\sigma_r({\bf c}') \geq {n \brack r} q^{m(r-d+1)}$.
\end{lemma}

\begin{proof}
Any codeword ${\bf c}' \in \mathcal{C}'$ can be expressed as ${\bf
c}' = c_{n-d+1} {\bf g}^{[n-d+1]} + c_{n-d+2} {\bf g}^{[n-d+2]} +
\ldots + c_{n-r} {\bf g}^{[n-r]}$, where $c_i \in \mathrm{GF}(q^m)$
for $n-d+1 \leq i \leq n-r$. If $c_{n-r} = 0$, then $\left(
\mathcal{C} + {\bf c}' \right) \subset \mathcal{D}$, where
$\mathcal{D}$ is the $(n,n-r,r+1)$ generalized Gabidulin code
generated by $\left({\bf g}^{[0]^T}, \,\, {\bf g}^{[1]^T},\,\,
\ldots, \,\, {\bf g}^{[n-r-1]^T} \right)^T$. Therefore
$\sigma_r({\bf c}') = 0$ if $c_{n-r} = 0$.

Denote the number of codewords of rank $r$ in $\mathcal{C} \oplus
\mathcal{C}'$ as $\tau_r$. Since $\bigcup_{{\bf c}' \in
\mathcal{C}'} \left( \mathcal{C} + {\bf c}' \right) = \mathcal{C}
\oplus \mathcal{C}'$, we have $\tau_r = \sum_{{\bf c}' \in
\mathcal{C}'} \sigma_r({\bf c}')$. Also, $\mathcal{C} \oplus
\mathcal{C}'$ forms an $(n,n-r+1,r)$ MRD code, and hence $\tau_r =
M(q,m,n,r,r) = {n \brack r} (q^m-1)$. Suppose that for all ${\bf c}'
\in \mathcal{C}'$, $\sigma_r({\bf c}') < {n \brack r} q^{m(r-d+1)}$.
Then $\tau_r = \sum_{{\bf c}': c_{n-r} \neq 0} \sigma_r({\bf c}') <
{n \brack r} (q^m-1)$, which contradicts $\tau_r = {n \brack r}
(q^m-1)$.
\end{proof}

Although Lemma~\ref{lemma:c'} proves the existence of a vector ${\bf
c}'$ for which the translate $\mathcal{C} + {\bf c}'$ has high
cardinality, it does not indicate how to choose ${\bf c}'$. For $d =
r+1$, it can be shown that all ${\bf c}' \in \mathcal{C}'$ satisfy
the bound, and that they all lead to optimal codes.

\begin{corollary}\label{cor:c'_r+1}
If $d=r+1$, then $\sigma_r({\bf c}') = {n \brack r}$ for all ${\bf
c}' \in \mathcal{C}'$.
\end{corollary}

\begin{proof}
First, by Proposition~\ref{prop:A_As}, $\sigma_r({\bf c}') \leq
\Ar(q,m,n,r+1,r) \leq \Ac(q,n,r,1) = {n \brack r}$ for all ${\bf c'}
\in \mathcal{C}'$. Suppose there exists ${\bf c}'$ such that
$\sigma_r({\bf c}') < {n \brack r}$. Then $\tau_r < {n \brack r}
(q^m-1)$, which contradicts $\tau_r = {n \brack r} (q^m-1)$.
\end{proof}

\begin{proposition}\label{prop:bound_Ar_gabidulin}
For all $q$, $1 \leq r < d \leq n \leq m$, $\Ar(q,m,n,d,r) \geq {n
\brack r} q^{n(r-d+1)}$, and a class of codes that satisfy this
bound can be constructed from Lemma~\ref{lemma:c'}.
\end{proposition}

\begin{proof}
The codewords of rank $r$ in a code considered in
Lemma~\ref{lemma:c'} form a CRC in $\mathrm{GF}(q)^{m \times n}$
with constant rank $r$, minimum distance $d$, and cardinality $\geq
{n \brack r} q^{m(r-d+1)}$. Therefore, $\Ar(q,m,n,d,r) \geq {n
\brack r} q^{m(r-d+1)}$. The proof is concluded by noting that
$\Ar(q,m,n,d,r) \geq \Ar(q,n,n,d,r) \geq {n \brack r}q^{n(r-d+1)}$.
\end{proof}

\begin{corollary}\label{cor:A_r+1}
For all $q$, $1 \leq r < n \leq m$, $\Ar(q,m,n,r+1,r) = {n \brack r}
= \Ac(q,n,r,1)$.
\end{corollary}
This can be shown by combining Propositions~\ref{prop:A_As} and
\ref{prop:bound_Ar_gabidulin}.
We note that ${n \brack r}$ is independent of $m$. We also remark
that the lower bound in Proposition~\ref{prop:bound_Ar_gabidulin} is
also trivial for $d$ approaching $2r$. Since the proof is only
partly constructive, computer search can be used to help find better
results for small parameter values.

By Proposition \ref{prop:A_As}, the lower bounds on $\Ar(q,m,n,d,r)$ derived in this section for $d > r$ can be viewed as lower bounds on the maximum cardinality of a corresponding CDC. Although in Corollary \ref{cor:A_r+1}, we obtain a tight bound for $d=r+1$, we remark that the bound in Proposition \ref{prop:bound_Ar_gabidulin} does not improve on the lower bounds on $\Ac(q,n,r,d-r)$ previously derived in the literature when $d > r+1$. However, the construction of good CDCs from CRCs is an interesting topic for future work.

\subsection{Asymptotic results}\label{sec:asymptotic}

We study the asymptotic behavior of CRCs using the following set of
normalized parameters: $\nu = \frac{n}{m}$, $\rho = \frac{r}{m}$,
and $\delta = \frac{d}{m}$. By definition, $0 \leq \rho, \delta \leq
\nu$, and since we assume $n \leq m$, $\nu \leq 1$. We consider the
asymptotic rate defined as $\ar(\nu,\delta,\rho) \df \lim_{m
\rightarrow \infty} \sup \left[\log_{q^{m^2}} \Ar(q,m,n,d,r)
\right]$.
We now investigate how $\Ar(q,m,n,d,r)$ behaves as the parameters
tend to infinity. Without loss of generality, we only consider the
case where $0 \leq \delta \leq \min\{ \nu, 2\rho\}$, since
$\ar(\nu,\delta,\rho) = 0$ for $\delta
> 2\rho$.

\begin{proposition}\label{prop:a}
For $0 \leq \delta \leq \rho$, $\ar(\nu,\delta,\rho) =
\rho(1+\nu-\rho) - \delta$. For $\rho \leq \delta$, we have to
distinguish three cases. First, for $2\rho \leq \nu$,
\begin{align}
	\nonumber
    &\max \left\{\frac{(1-\rho)(\nu-\rho)}{1+\nu-2\rho}(2\rho-\delta),
    \rho(2\nu-\rho) - \nu\delta \right\}\\
	\label{eq:a_p1}
    &\leq \ar(\nu,\delta,\rho)\leq (\nu-\rho)(2\rho-\delta).
\end{align}
Second, for $\nu \leq 2\rho \leq 1$,
\begin{align}
	\nonumber
    &\max \left\{ \rho(1-\rho)(\nu-\delta), \rho(2\nu-\rho) - \nu\delta \right\}\\
    \label{eq:a_p2}
    &\leq \ar(\nu,\delta,\rho) \leq \rho(\nu-\delta).
\end{align}
Third, for $2\rho \geq 1$,
\begin{align}
    \nonumber
    &\max \left\{\frac{\rho}{2}(1 + \nu - 2 \rho - \delta), \rho(2\nu-\rho) - \nu\delta, 0 \right\}\\
    \label{eq:a_p3}
    &\leq \ar(\nu,\delta,\rho) \leq \rho(\nu-\delta).
\end{align}
\end{proposition}

The proof of Proposition~\ref{prop:a} is given in
Appendix~\ref{app:prop:a}.

Proposition~\ref{prop:bound_Mdr} indicates that the codewords of a
given rank in a linear MRD code form asymptotically optimal CRCs. In
particular, Proposition~\ref{prop:a} shows that the set of codewords
with rank $n$ in an $(n,n-d+1,d)$ linear MRD code constitutes a CRC
of rank $n$ and asymptotic rate of $\nu-\delta$, which is equal to
the asymptotic rate of an optimal rank metric code
\cite{gadouleau_it08_covering}.

We can split the range of $\delta$ into two regions: when $\delta
\leq \rho$, the asymptotic rate of CRCs is determined due to the
construction of good CRCs when $d \leq r$; when $\delta \geq \rho$,
we only have bounds on the asymptotic rate of CRCs. Also, the lower
bounds based on the connection between CDCs and CRCs (the first
lower bound in the LHS of (\ref{eq:a_p1}), (\ref{eq:a_p2}), and
(\ref{eq:a_p3}) are tighter for $2\rho \leq \nu$ and on the other
hand become trivial for $\rho$ approaching $1$. The bounds on
$\ar(\nu,\delta,\rho)$ in the three cases in (\ref{eq:a_p1}), (\ref{eq:a_p2}), and
(\ref{eq:a_p3}) are illustrated in
Figures~\ref{fig:a_0.2}, \ref{fig:a_0.4}, and \ref{fig:a_0.6}
for $\nu =3/4$ and $\rho=1/5, 2/5, 3/5$, respectively.

%
%
%


\begin{figure}
\begin{center}
\includegraphics[scale = 0.5]{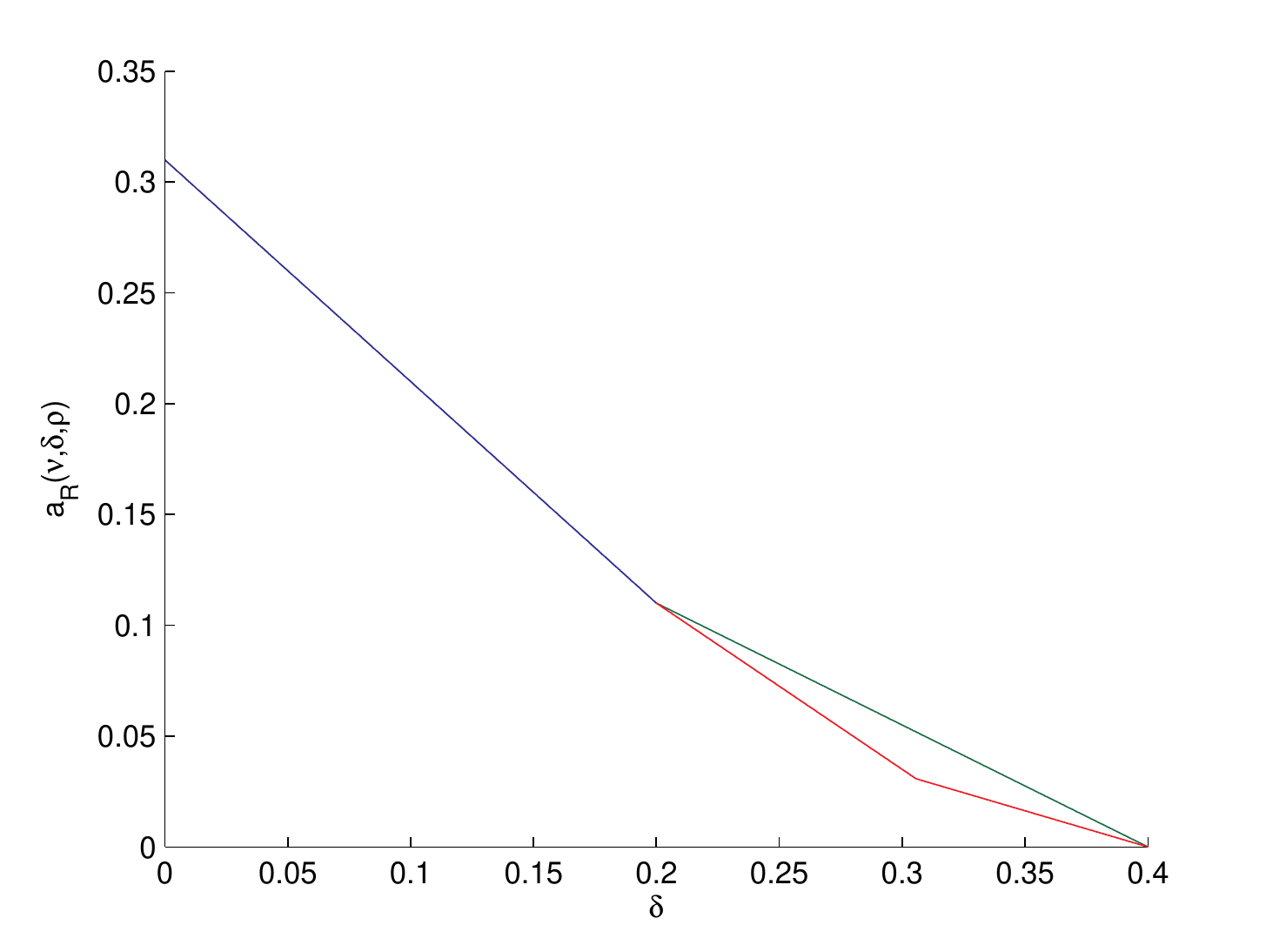}
\end{center}
\caption{Asymptotic bounds on the maximal rate of a CRC as a
function of $\delta$, with $\nu = 3/4$ and $\rho = 1/5$.}\label{fig:a_0.2}
\end{figure}

\begin{figure}
\begin{center}
\includegraphics[scale = 0.5]{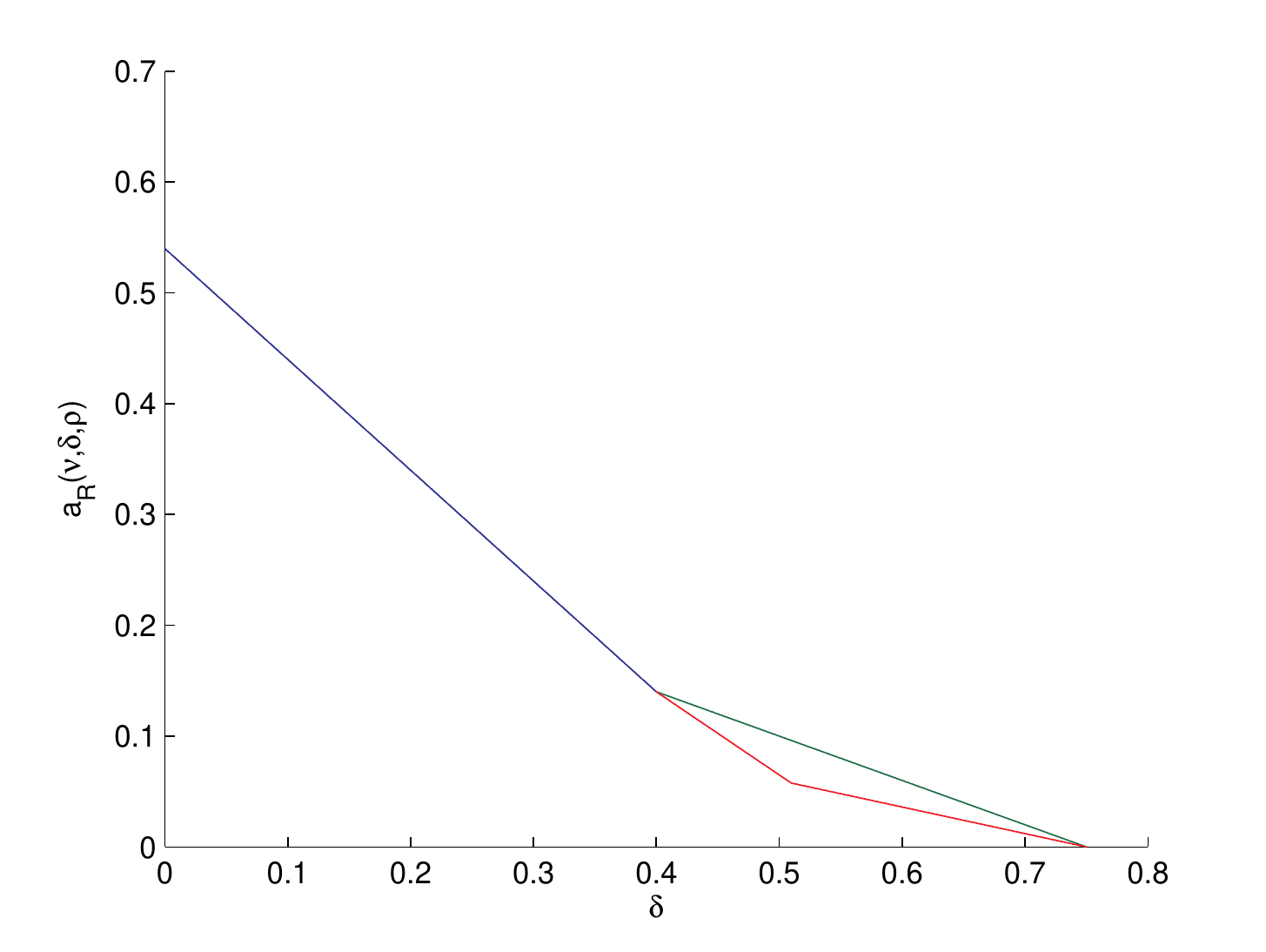}
\end{center}
\caption{Asymptotic bounds on the maximal rate of a CRC as a
function of $\delta$, with $\nu = 3/4$ and $\rho = 2/5$.}\label{fig:a_0.4}
\end{figure}

\begin{figure}
\begin{center}
\includegraphics[scale = 0.5]{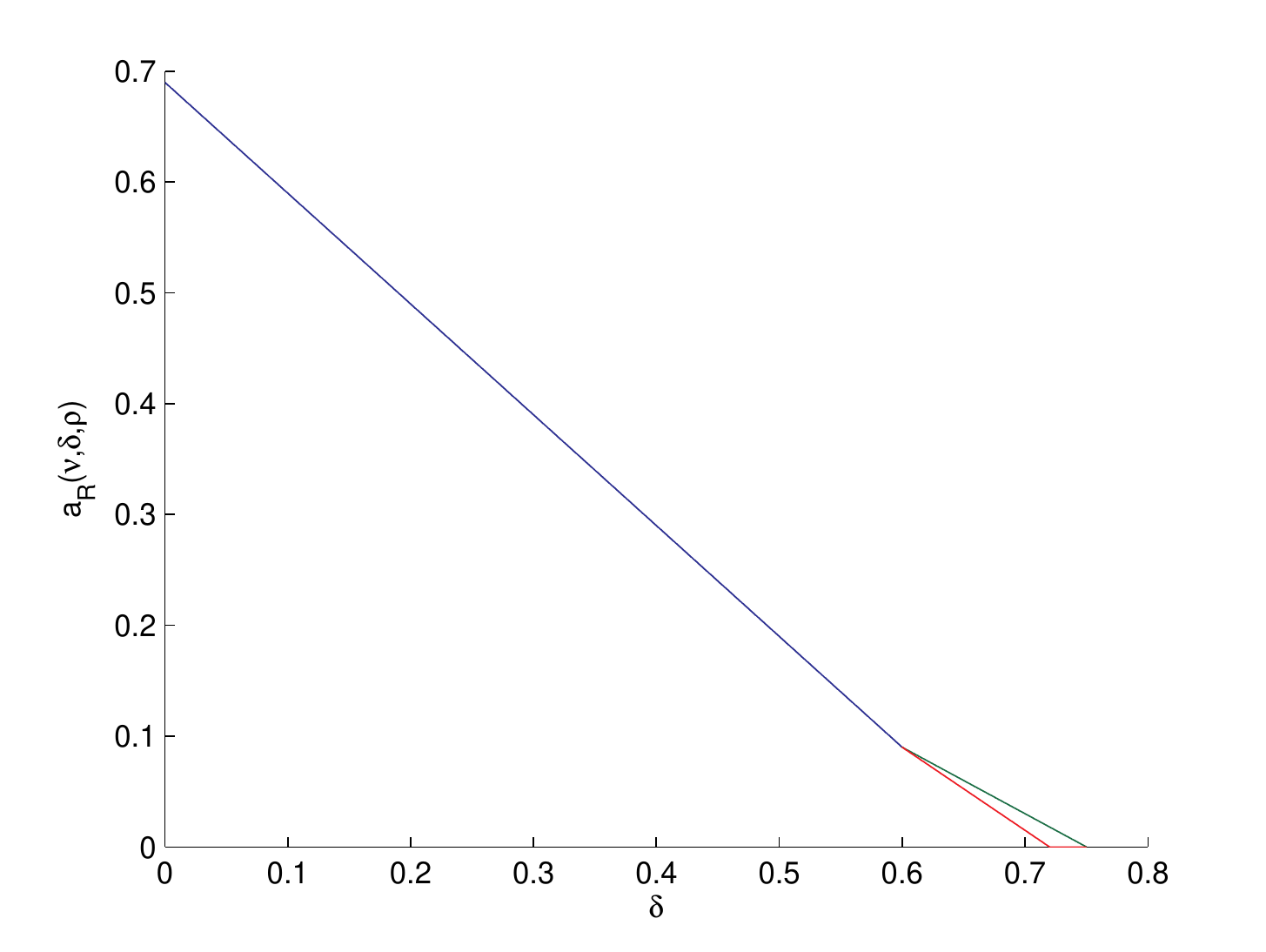}
\end{center}
\caption{Asymptotic bounds on the maximal rate of a CRC as a
function of $\delta$, with $\nu = 3/4$ and $\rho = 3/5$.}\label{fig:a_0.6}
\end{figure}

\section{Conclusion}
Rank metric codes and CDCs have been considered
for error control in noncoherent random linear network coding. It has been shown that these two
classes of codes are related by the lifting operation, which turns
an optimal rank metric code into a nearly optimal constant-dimension
code. However, liftings of rank metric codes are not optimal constant-dimension codes. In this paper, we first established a novel connection between CRCs and CDCs, by showing that optimal
CRCs over matrices with sufficiently many rows lead
to optimal CDCs with a related minimum injection distance. In comparison to previously proposed constructions of CDCs, our construction based on CRCs guarantees the optimality of CDCs, and hence is a promising approach. Despite previous works on rank metric codes in general, CRCs have received little attention in the literature. We hence investigated the properties of CRCs, derived
bounds on their cardinalities, and proposed explicit constructions of
CRCs in some cases. Although we have not been able to propose constructions of optimal CRCs in all cases, we hope our novel connection between CRCs and CDCs and investigation of CRCs can lead to constructions of optimal CDCs, which is the topic of our future work.

\section{Acknowledgment}
The authors are grateful to the anonymous reviewers and the Associate Editor Dr.~Ludo Tolhuizen for their constructive comments, which have resulted in improvements in both the results and the presentation of the paper. In particular, comments by one of the reviewers have slightly improved the results in Theorem~\ref{th:x_y_E_F} and Propositions~\ref{prop:E(C)}, \ref{prop:C_Gamma_Delta}, \ref{prop:A_As}, and \ref{prop:a}.

\appendix

\subsection{Proof of Proposition~\ref{prop:bassalygo_s} (Bassalygo-Elias bound for rank metric codes)}
\label{app:prop:bassalygo_s}

\begin{proof}
For all ${\bf X} \in \mathrm{GF}(q)^{l \times k}$ with rank $s$ and
${\bf C} \in \mathcal{C}$, we define $f_r({\bf X}, {\bf C}) = 1$ if
$\dr({\bf X}, {\bf c}) = r$ and $f_r({\bf X}, {\bf C}) = 0$
otherwise. Note that $\sum_{{\bf X} : \rk({\bf X}) = s} f_r({\bf X},
{\bf C}) = \Jr(q,l,k,s,r,\rk({\bf C}))$ for all ${\bf C} \in
\mathcal{C}$ and $\sum_{{\bf c} \in C} f_r({\bf X}, {\bf C}) =
\left| \left\{ {\bf Y} \in C - {\bf X} : \rk({\bf Y}) = r \right\}
\right| \leq \Ar(q,l,k,d,r)$ for all ${\bf X} \in \mathrm{GF}(q)^{l
\times k}$. We obtain
\begin{align}
    \label{eq:b_extended1_s}
    &\sum_{{\bf C} \in \mathcal{C}} \sum_{{\bf X} : \rk({\bf X}) = s}
    f_r({\bf X}, {\bf C}) = \sum_{i=0}^n A_i \Jr(q,l,k,s,r,i),\\
    \label{eq:b_extended3_s}
    &\sum_{{\bf X}: \rk({\bf X}) = s} \sum_{{\bf C} \in \mathcal{C}}
    f_r({\bf X}, {\bf C})\leq \Nr(q,l,k,s) \Ar(q,l,k,d,r).
\end{align}
Combining~(\ref{eq:b_extended1_s}) and~(\ref{eq:b_extended3_s}), we
obtain
\begin{equation}\label{eq:bassalygo_s2}
    \Ar(q,l,k,d,r) \geq \frac{\sum_{i=0}^n A_i
    \Jr(q,l,k,s,r,i)}{\Nr(q,l,k,s)}.
\end{equation}

Suppose $d > r+1$. For all ${\bf C} \in \mathcal{C}$, let us denote
the set of matrices with rank $s$ at distance at most $d-r-1$ from
${\bf C}$ as $S_{\bf C}$, and $S \df \bigcup_{{\bf C} \in
\mathcal{C}} S_{\bf C}$. For ${\bf X} \in S_{\bf C}$, we have
$\dr({\bf X}, {\bf C}) \leq d-r-1 < r$. We have for ${\bf C}' \in
\mathcal{C}$ and ${\bf C}' \neq {\bf C}$, $\dr({\bf X}, {\bf C}')
\geq \dr({\bf C}, {\bf C}') - \dr({\bf X}, {\bf C}) \geq r+1$; and
hence $f_r({\bf X}, {\bf C}') = 0$ for all ${\bf C}' \in
\mathcal{C}$. Therefore, $\sum_{{\bf C} \in \mathcal{C}} f_r({\bf
X}, {\bf C}) = 0$ for all ${\bf X} \in S$ and
\begin{align} \nonumber
    &\sum_{{\bf X} : \rk({\bf X}) = s} \sum_{{\bf C} \in \mathcal{C}}
    f_r({\bf X}, {\bf C})\\
    \nonumber
    &= \sum_{{\bf X} \in S} \sum_{{\bf C} \in C} f_r({\bf X}, {\bf C})
    + \sum_{\stackrel{{\bf X} \notin S}{\rk({\bf X}) = s}}
    \sum_{{\bf C} \in \mathcal{C}} f_r({\bf X}, {\bf C})\\
    \label{eq:b_extended2_s}
    &\leq \left[\Nr(q,l,k,s)- |S|\right]\Ar(q,l,k,d,r).
\end{align}

Since $d-r-1 < \frac{d}{2}$, the balls with radius $d-r-1$ around
the codewords are disjoint and hence $|S| = \sum_{i=0}^n A_i
\sum_{t=0}^{d-r-1} \Jr(q,l,k,s,t,i)$.
Combining~(\ref{eq:b_extended1_s}) and~(\ref{eq:b_extended2_s}), we
obtain
\begin{align}
	\nonumber
    &\Ar(q,l,k,d,r)\\
    \label{eq:bassalygo_2}
    &\geq \frac{\sum_{i=0}^n A_i
    \Jr(q,l,k,s,r,i)}{\Nr(q,l,k,s) - \sum_{i=0}^n A_i \sum_{t=0}^{d-r-1}
    \Jr(q,l,k,s,t,i)}.
\end{align}
Note that (\ref{eq:bassalygo_s2}) and (\ref{eq:bassalygo_2}) both
hold for any $s$ and rank spectrum $\left\{A_i\right\}$.
Furthermore, since $\Ar(q,l,k,d,r)$ is a non-decreasing function of
$l$ and $k$, $\Ar(q,m,n,d,r) \geq \Ar(q,l,k,d,r)$ for all
$\max\{r,d\} \leq k \leq n$ and $k \leq l \leq m$. Thus, we have
(\ref{eq:bassalygo_s}) and (\ref{eq:bassalygo_extended_s}).
\end{proof}

\subsection{Proof of Proposition~\ref{prop:bound_Mdr}}
\label{app:prop:bound_Mdr}

\begin{proof}
By Proposition~\ref{prop:singleton_johnson}, we obtain
$\Ar(q,m,m,d,m) \leq \alpha(m,m-d+1)$ for $r=n=m$ and
$\Ar(q,m,n,d,r) \leq {n \brack r}\alpha(m,r-d+1) < {n \brack
r}q^{m(r-d+1)}$ otherwise. We now derive lower bounds on
$M(q,m,n,d,r)$. Again, $M(q,m,n,d,r) = {n \brack r} \sum_{j=d}^r
(-1)^j \mu_j$ where $\mu_j > \mu_{j-1}$ for $d+1 \leq j \leq r$.
Therefore, when needed, we shall only consider the last terms in the
summation.

First, $M(q,m,m,m-1,m) = (q^{2m} - 1) - \frac{q^m-1}{q-1}(q^m-1) >
\frac{q-2}{q-1}(q^{2m}-1) > \frac{q-2}{q-1} \alpha(m,2)$, which
leads to (\ref{eq:B_2}). For $q=2$, $M(2,m,m,m-1,m) = 2(2^m-1) =
(2^{m-1}-1)^{-1}\alpha(m,2)$, which results in (\ref{eq:B_1}).
Second, when $r=n=m$ and $d=m-2$,
\begin{align*}
    &M(q,m,m,m-2,m)\\
    &= (q^{3m}-1) - \frac{\alpha(m,1)}{q-1} (q^{2m}-1) +
    \frac{\alpha(m,2)}{(q^2-1)(q-1)} (q^m-1)\\
    &> \frac{q-2}{q-1} \alpha(m,1)(q^{2m}-1) + \frac{1}{(q^2-1)(q-1)}
    \alpha(m,2)(q^m-1)\\
    &> \frac{(q^2-1)(q-2) + 1}{(q^2-1)(q-1)} \alpha(m,1),
\end{align*}
which leads to (\ref{eq:B_3}). Third, when $r=n=m$ and $d<m-2$, by
considering the last four terms in the summation, we obtain
\begin{align*}
    &M(q,m,m,d,m)\\
    &> (q^{m(m-d+1)}-1) - \frac{\alpha(m,1)}{q-1} (q^{m(m-d)}-1)\\
    &+ \frac{\alpha(m,2)}{(q^2-1)(q-1)} (q^{m(m-d-1)}-1)\\
    &- \frac{\alpha(m,3)}{(q^3-1)(q^2-1)(q-1)} (q^{m(m-d-2)}-1)\\
    &> \left\{ \frac{q-2}{q-1} + \frac{q^3-2}{(q^3-1)(q^2-1)(q-1)} \right\} \alpha(m,m-d+1),
\end{align*}
which results in (\ref{eq:B_4}). Fourth, when $d < r < m$, by
considering the last two terms in the summation, we obtain
\begin{align*}
    &M(q,m,n,d,r)\\
    &\geq {n \brack r} \left( (q^{m(r-d+1)} - 1) - {r \brack 1}(q^{m(r-d)} - 1) \right) \\
    &\geq {n \brack r} \left(q^{m(r-d+1)} - 1 -q^{m(r-d) + r} + q^r\right)\\
    &\geq {n \brack r} q^{m(r-d+1)} (1-q^{r-m}).
\end{align*}
Therefore, since $r < m$, $B(q,m,n,d,r) < (1-q^{r-m})^{-1} \leq
\frac{q}{q-1}$, which leads to (\ref{eq:B_5}).
\end{proof}

\subsection{Proof of Proposition~\ref{prop:a}} \label{app:prop:a}

\begin{proof}
We first derive a lower bound on $\ar(\nu,\delta,\rho)$. For $d \leq
r$, Proposition~\ref{prop:Mdr} yields $\Ar(d,r) \geq q^{r(n-r)
+ m(r-d)}$, which asymptotically becomes $\ar(\nu,\delta,\rho) \geq
\rho(1+\nu-\rho) - \delta$ for $\delta \leq \rho$. Similarly, for $d
> r$, Proposition~\ref{prop:bound_Ar_gabidulin} yields
$\Ar(q,m,n,d,r) \geq q^{r(n-r) + n(r-d+1)}$, which asymptotically
becomes $\ar(\nu,\delta,\rho) \geq \rho(2\nu-\rho) - \nu \delta$ for
$\delta \geq \rho$.

Proposition~\ref{prop:A_As} and (\ref{eq:bounds_Ac}) yield $\log_q
\Ar(d,r) \geq \min\{(n-r)(2r-d-p+1), (m-r)(p+1) \}$ for $d
> r$ and $2r \leq n$. Treating the two terms as functions and
assuming that $p$ is real, the lower bound is maximized when $p =
\frac{(n-r)(2r-d+1) - m+r}{m+n-2r}$. Using $p = \left \lfloor
\frac{(n-r)(2r-d+1) - m+r}{m+n-2r}\right\rfloor$, asymptotically we
obtain $\ar(\nu,\delta,\rho) \geq
\frac{(1-\rho)(\nu-\rho)}{1+\nu-2\rho}(2\rho-\delta)$ for $2\rho
\leq \nu$.

For $d > r$ and $n \leq 2r \leq m$, Proposition~\ref{prop:A_As} and
(\ref{eq:bounds_Ac}) lead to $\log_q \Ar(d,r) \geq \min\{
r(n-d-p+1), (m-r)(p+1)\}$. After maximizing this expression over
$p$, we asymptotically obtain $\ar(\nu,\delta,\rho) \geq
\rho(1-\rho)(\nu-\delta)$ for $\nu \leq 2\rho \leq 1$.

For $d > r$ and $2r \geq m$, Proposition~\ref{prop:A_As} and
(\ref{eq:bounds_Ac}) lead to $\log_q \Ar(d,r) \geq \min\{
r(n-d-p+1), r(m-2r+p+1)\}$. After maximizing this expression over
$p$, we asymptotically obtain $\ar(\nu,\delta,\rho) \geq
\frac{\rho}{2}(1 + \nu - 2\rho - \delta)$ for $2\rho \geq 1$.

We now derive an upper bound on $\ar(\nu,\delta,\rho)$. First,
Proposition~\ref{prop:singleton_johnson} gives $\Ar(d,r) < {n
\brack r} q^{m(r-d+1)} < K_q^{-1} q^{r(n-r) + m(r-d+1)}$ for $d \leq
r$, which asymptotically becomes $\ar(\nu,\delta,\rho) \leq
\rho(1+\nu-\rho) - \delta$ for $\rho \geq \delta$. Second, by
Proposition~\ref{prop:A_As}, we obtain $\ar(\nu, \delta, \rho) \leq
\lim_{m \rightarrow \infty} \sup \left[ \log_{q^{m^2}}
\Ac(q,n,r,d-r) \right] = \min\{ (\nu-\rho)(2\rho-\delta),
\rho(\nu-\delta) \}$ for $\rho \leq \delta \leq \min\{2\rho,\nu\}$.
\end{proof}


\bibliographystyle {IEEEtr}

\end{document}